\documentclass[runningheads]{llncs}

\usepackage{multirow}
\usepackage{xspace}
\usepackage{paralist}
\usepackage[utf8]{inputenc}
\usepackage[T1]{fontenc}
\usepackage{amsmath} 
\usepackage{amssymb}
\usepackage{graphicx}
\usepackage{wrapfig}
\usepackage{complexity}
\usepackage{microtype}
\usepackage{cite}
\usepackage{subfigure}
\usepackage{booktabs}
\usepackage{mathtools}

\usepackage[vlined,linesnumbered,ruled]{algorithm2e}
\DontPrintSemicolon
\usepackage{times}

\usepackage{booktabs}

\usepackage{todonotes}
\usepackage{hyperref}
\usepackage{url}

\DeclareMathOperator{\lw}{\lambda} \DeclareMathOperator{\dist}{d}
\DeclareMathOperator{\argmax}{argmax}

\DeclareMathOperator{\length}{len}
\DeclareMathOperator{\maxlabel}{opt-cand}

\newcommand{\ml}{\ensuremath{\maxlabel}\xspace}
 
\newcommand{\OPT}{\ensuremath{\mathcal L}\xspace}
\newcommand{\MaxTotalCovering}
{\textsc{Max\-I\-den\-ti\-fied\-Roads}\xspace}

\newcommand{\OPTCAND}{\texttt{Opt\-Can\-di\-da\-te}\xspace}

\newcommand{\weighting}{\ensuremath{\overline \omega}\xspace}
\newcommand{\weight}{\ensuremath{\omega}\xspace}

\title{Label Placement in Road Maps}

\author{Andreas Gemsa \and Benjamin Niedermann 
  \and Martin N\"ollenburg 
} \institute{Institute of Theoretical Informatics, Karlsruhe Institute
  of Technology, Germany }

\begin{document}
\maketitle

\begin{abstract} 
	
  A road map can be interpreted as a graph embedded in the plane, in which each
  vertex corresponds to a road junction and each edge to a particular road
  section. We consider the cartographic problem to place non-overlapping road
  labels along the edges so that as many road sections as possible are
  identified by their name, i.e., covered by a label. We show that this is
  \NP-hard in general, but the problem can be solved in polynomial time if the
  road map is an embedded tree. 

\end{abstract}

\section{Introduction}

Map labeling is a well-known cartographic problem in computational
geometry~\cite[Chapter 58.3.1]{overview},\cite{bibliography}. Depending on the type of map
features, one can distinguish labeling of \emph{points}, \emph{lines},
and \emph{areas}. Common cartographic quality criteria are that labels
must be disjoint and clearly identify their respective map
features~\cite{criteria}. Most of the previous work concerns point labeling, while
labeling line and area features received considerably less
attention. In this paper we address labeling linear
features, namely roads in a road map.

Geometrically, a \emph{road map} is the representation of a \emph{road
  graph} $G$ as an arrangement of fat curves in the plane $\mathbb
R^2$. Each \emph{road} is a connected subgraph of $G$ (typically a
simple path) and each edge belongs to exactly one road. Roads may
intersect each other in \emph{junctions}, the vertices of $G$, and we
denote an edge connecting two junctions as a \emph{road section}. In
road labeling the task is to place the road names inside the fat
curves so that the road sections are identified unambiguously, see
Fig.~\ref{fig:goodbad}.

Chiri{\'e}~\cite{street-name-placement} presented a set of rules and
quality criteria for label placement in road maps based on interviews
with cartographers.  This includes that (C1) labels are placed inside
and parallel to the road shapes, (C2) every road section between two
junctions should be clearly identified, and (C3) no two road labels
may intersect.  Further, he gave a mathematical description for
labeling a single road and introduced a heuristic for sequentially
labeling all roads in the map. 
Imhof's foundational cartographic work on label positioning in maps lists very similar quality criteria~\cite{imhof}.
Edmondson et al.~\cite{Edmondson96} took an algorithmic
perspective on labeling a single linear feature (such as a river). 
While Edmondson et al.\ considered \emph{non-bent} labels, Wolff et
al.~\cite{wkksa-seahq-00} introduced an algorithm for single linear
feature that places labels following the curvature of the linear
feature.  Strijk \cite{strijk2001} considered static road labeling
with embedded labels and presented a heuristic for selecting
non-overlapping labels out of a set of label candidates.  Seibert and
Unger~\cite{labelingManhattan} considered grid-shaped road networks.
They showed that in those networks it is \NP-complete and \APX-hard to
decide whether for every road at least one label can be placed. Yet,
Neyer and Wagner~\cite{downtownLabeling} introduced a practically
efficient algorithm that finds such a grid labeling if possible.
Maass and Döllner~\cite{Maass07} presented a heuristic for labeling the
roads of an interactive 3D map with objects (such as buildings).
Apart from label-label overlaps, they also resolve label-object
occlusions. Vaaraniem et al.~\cite{Vaaraniemi12} used a force-based
labeling algorithm for 2D and 3D scenes including road label placement.

\paragraph{Contribution.} While in grid-shaped road networks it is sufficient to place a single
label per road to clearly identify all its road sections, this is not
the case in general road networks.  Consider the example in
Fig.~\ref{fig:goodbad}. In Fig.~\ref{fig:goodbad}a), it is not obvious
whether the orange road section in the center belongs to \emph{Knuth
  St.}\ or to~\emph{Turing St}. Simply maximizing the number of placed
labels, as often done for labeling point features, can cause undesired
effects like unnamed roads or clumsy label placements (e.g., around
\emph{Dijkstra St.}\ and \emph{Hamming St.} in
Fig.~\ref{fig:goodbad}a)).  Therefore, in contrast to Seibert and
Unger~\cite{labelingManhattan}, we aim for maximizing the number
of~\emph{identified} road sections, i.e., road sections that can be
clearly assigned to labels; see Fig.~\ref{fig:goodbad}b).  

Based on criteria (C1)--(C3) we introduce a new and versatile model for road
labeling in Section~\ref{sec:model}. In Section~\ref{sec:np-hardness} we show
that the problem of maximizing the number of identified road sections
is~\NP-hard for general road graphs, even if each road is a path. For the
special case that the road graph is a tree, we present a polynomial-time
algorithm in Section~\ref{sec:tree-algorithm}. This special case is not only of
theoretical interest, but our algorithm in fact provides a very useful
subroutine in exact or heuristic algorithms for labeling general road graphs.
Our initial experiments, sketched in Section~\ref{sec:appendix:experiments},
show that real-world road networks decompose into small subgraphs, a large
fraction of which (more than 85.1\%) are actually trees, and thus can be labeled optimally by our algorithm.

\begin{figure}[t]
  \centering
  \includegraphics[page=3,scale=0.95]{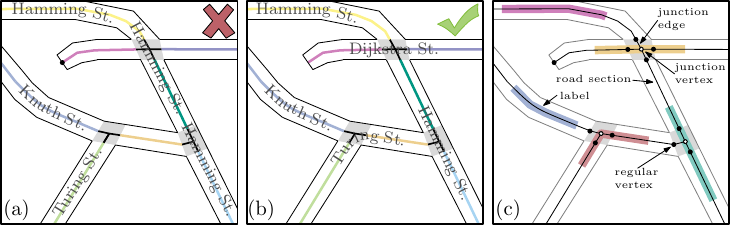}
  \caption{\small a--b): Two ways to label the same road network. Each
    road section has its own color. Junctions are marked gray. Fig.~b)
    identifies all road sections. c)~Illustration of the
    road graph and relevant terms.}
  \label{fig:goodbad}
\end{figure}

\section{Preliminaries}\label{sec:model}

As argued above, a road map is a collection of fat curves in the
plane, each representing a particular piece of a named road. If two
(or more) such curves intersect, they form junctions. A \emph{road
  label} is again a fat curve (the bounding shape of the road name)
that is contained in and parallel to the fat curve representing its
road. We observe that labels of different roads can intersect only
within junctions and that the actual width of the curves is
irrelevant, except for defining the shape and size of the junctions.
These observations allow us to define the following more abstract but
 equivalent road map model.

A \emph{road map}~$\mathcal M$ is a planar \emph{road graph} $G=(V,E)$ together
with a planar embedding $\E(G)$, which can be thought of as the
geometric representation of the road axes as thin curves; see Fig~\ref{fig:goodbad}c).  We denote
the number of vertices of $G$ by $n$, and the number of edges by
$m$. Observe that since $G$ is planar $m = O(n)$.  Each edge $e \in E$
is either a \emph{road section}, which is not part of a junction, or a
\emph{junction edge}, which is part of a junction.  Each vertex $v \in
V$ is either a \emph{junction vertex} incident only to junction edges,
or a \emph{regular vertex} incident to one road section and at most one
junction edge, which implies that each regular vertex has degree
at most two.  A junction vertex $v$ and its incident edges are
denoted as a \emph{junction}.  The edge set $E$ decomposes into a set
$\mathcal R$ of edge-disjoint \emph{roads}, where each road $R \in
\mathcal R$ induces a connected subgraph of $G$. Without loss of generality
we assume no two road sections~$G$ are incident to the same vertex.
Thus, a road decomposes into road sections, separated by junction
vertices and their incident junction edges. In realistic road networks the
number of roads connected passing through a junction is small and does not depend on
the size of the road network. We therefore assume that each vertex in~$G$
has constant degree. We assume that each road~$R\in \mathcal R$
has a name whose length we denote by~$\lw(R)$. 

For simplicity, we identify the embedding $\E(G)$ with the points in
the plane covered by $\E(G)$, i.e.~$\E(G)\subseteq \mathbb{R}^2$.
We also use $\E(v)$, $\E(e)$, and
$\E(R)$ to denote the embeddings of a vertex $v$, an edge $e$, and a
road $R$.

We model a label as a simple open curve $\ell\colon [0,1]\to \E(G)$ in~$\E(G)$.
Unless mentioned otherwise, we consider a curve~$\ell$ always to be
simple and open, i.e., $\ell$ has no self-intersections and its end
points do not coincide. In order to ease the description, we identify
a curve $\ell$ in $\E(G)$ with its image, i.e.,~$\ell$ denotes the
set~$\{\ell(t)\in \E(G)\mid t\in[0,1]\}$.  The start point of $\ell$
is denoted as the \emph{head} $h(\ell)$ and the endpoint as the
\emph{tail} $t(\ell)$. The length of~$\ell$ is denoted by
$\length(\ell)$.  The curve $\ell$ \emph{identifies} a road section
$r$ if $\ell \cap \E(r) \ne \emptyset$. For a set $\mathcal L$ of
curves $\weight(\mathcal L)$ is the number of road sections that are
identified by the curves in~$\mathcal L$. For a single curve~$\ell$ we
use~$\weight(\ell)$ instead of~$\weight(\{\ell\})$. For two
curves~$\ell_1$ and $\ell_2$ it is not necessarily true that
$\weight(\{\ell_1,\ell_2\})=\weight(\ell_1)+\weight(\ell_2)$, because
they may identify the same road section twice.

A \emph{label} $\ell$ for a road $R$ is a curve $\ell \subseteq
\E(R)$ of length $\lw(R)$ whose endpoints must lie on road sections
and not on junction edges or junction vertices. Requiring that labels
end on road sections avoids ambiguous  placement of labels in junctions
where it is unclear how the road passes through it. 
A \emph{labeling} $\mathcal L$ for a road map with road set $\mathcal
R$ is a set of mutually non-overlapping labels, where
we say that two labels $\ell$ and $\ell'$ \emph{overlap} if they
intersect in a point that is not their respective head or tail.

Following the cartographic quality criteria~(C1)--(C3), our goal is to find a
labeling $\mathcal L$ that maximizes the number of identified road sections, i.e., for any labeling $\mathcal L'$ we have
$\weight(\mathcal L')\leq \weight(\mathcal L)$.  We call this
problem~\MaxTotalCovering.

Note that assuming the road graph~$G$ to be planar is not a
restriction in practice. Consider for example a road section~$r$ that
overpasses another road section~$r'$, i.e., $r$ is a bridge over~$r'$,
or $r'$ is a tunnel underneath~$r$. In order to avoid overlaps between
labels placed on $r$ and $r'$, we either can model the intersection
of~$r$ and $r'$ as a regular crossing of two roads or we split~$r'$ in
smaller road sections that do not cross~$r$. In both cases the
corresponding road graph becomes planar. In the latter case we may
obtain more independent roads created by chopping~$r'$ into smaller
pieces.

\section{Computational Complexity}\label{sec:np-hardness}

We first study the computational complexity of road labeling and prove
\NP-hardness of \MaxTotalCovering in the following sense.

\begin{theorem}\label{thm:npc}
  For a given road map $\mathcal M$ and an integer $K$ it is \NP-hard to decide if in total at least $K$ road sections can be identified.
\end{theorem}

\begin{proof}
  We perform a reduction from the \NP-complete \textsc{planar monotone
    3-Sat} problem~\cite{l-pftu-82}. An instance of \textsc{planar
    monotone 3-Sat} is a Boolean formula $\varphi$ with $n$ variables
  and $m$ clauses (disjunctions of at most three literals) that
  satisfies the following additional requirements: (i) $\varphi$ is
  \emph{monotone}, i.e., every clause contains either only positive
  literals or only negative literals and (ii) the induced
  variable-clause graph $H_\varphi$ of $\varphi$ is planar and can be
  embedded in the plane with all variable vertices on a horizontal
  line, all positive clause vertices on one side of the line, all
  negative clauses on the other side of the line, and the edges drawn
  as rectilinear curves connecting clauses and contained variables on
  their respective side of the line.
 \begin{figure}[t]
   \begin{center}
     \includegraphics[page=5,scale=0.9]{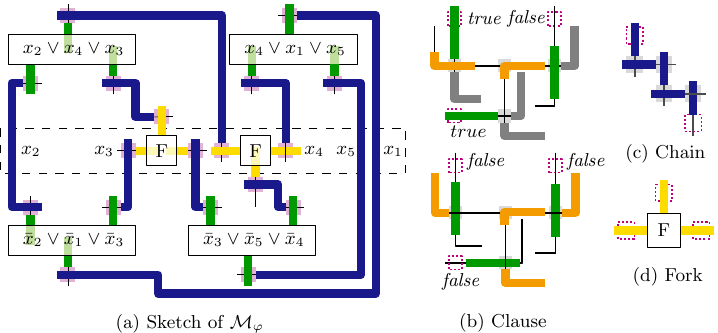}
   \end{center}
   \caption{\small Illustration of \NP-hardness proof. (a)~3-Sat
     formula $\varphi =( x_4 \vee x_1 \vee x_5) \wedge (x_2 \vee x_4
     \vee x_3) \wedge (\bar x_2 \vee \bar x_1 \vee \bar x_3) \wedge (
     \bar x_3 \vee \bar x_5 \vee \bar x_4)$ represented as road
     graph $\mathcal M_\varphi$. Truth assignment is $x_1=\mathit{true}$,
     $x_2=\mathit{true}$, $x_3=\mathit{false}$, $x_4=\mathit{false}$
     and $x_5=\mathit{false}$. (b)~Clause gadget in two states. (c)~The chain
     is the basic building block for the proof. (d)~Schematized fork gadget.}
   \label{fig:gadgets}
 \end{figure}
  We construct a road map $\mathcal M_\varphi$ that mimics the shape
  of the above embedding of $H_\varphi$ by defining variable and
  clause gadgets, which simulate the assignment of truth values to
  variables and the evaluation of the clauses. We refer to
  Fig.~\ref{fig:gadgets} for a sketch of the construction.
	
  \textit{Chain Gadget.} The basic building block is the
  \emph{chain gadget}, which consists of an alternating sequence of
  equally long horizontal and vertical roads with identical label
  lengths that intersect their respective neighbors in the sequence
  and form junctions with them as indicated in
  Fig.~\ref{fig:gadgets}c). Assume that the chain consists of $k\ge 3$
  roads. Then each road except the first and last one decomposes into
  three road sections split by two junctions, a longer central section
  and two short end sections; the first and last road consist of only
  two road sections, a short one and a long one, separated by one
  junction. (These two roads will later be connected to other gadgets;
  indicated by dotted squares in Fig.~\ref{fig:gadgets}c).) The label
  length and distance between junctions is chosen so that for each
  road either the central and one end section is identified, or no
  section at all is identified. For the first and last road, both
  sections are identified if the junction is covered and otherwise
  only the long section can be identified. We have $k$ roads and
  $k-1$ junctions. Each label must block a junction, if it identifies
  two sections. So the best possible configuration blocks all
  junctions and identifies $2(k-1) + 1 = 2k - 1$ road sections.
		
  The chain gadget has exactly two states, in which $2k-1$
  road sections are identified. Either the label of the first road
  does not block a junction and identifies a single section and all
  subsequent roads have their label cover the junction with the
  preceding road in the sequence, or the label of the last road does
  not block a junction and all other roads have their label cover the
  junction with the successive road in the sequence. In any other
  configuration there is at least one road without any identified
  section and thus at most $2k-2$ sections are identified. We use
  the two optimal states of the gadget to represent and transmit the
  values \emph{true} and \emph{false} from one end to the other.

  \textit{Fork Gadget.} The \emph{fork gadget} allows to split the
  value represented in one chain into two chains, which is needed to
  transmit the truth value of a variable into multiple clauses. To
  that end it connects to an end road of three chain gadgets by
  sharing junctions.

  \begin{figure}[t]
    \centering \subfigure[]{
      \includegraphics[page=11]{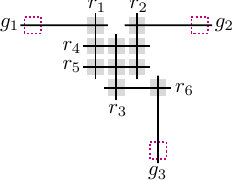} \label{fig:fork:structure}}
    \subfigure[]{
      \includegraphics[page=12]{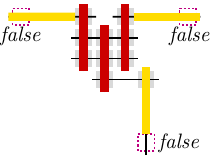} \label{fig:fork:false}}
    \subfigure[]{
      \includegraphics[page=13]{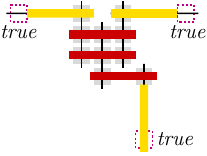} \label{fig:fork:true}}
    \caption{Illustration of the fork
      gadget. \protect\subref{fig:fork:structure} Structure of the
      fork gadget. \protect\subref{fig:fork:true} Configuration
      transmitting the value
      \emph{false}. \protect\subref{fig:fork:false} Configuration
      transmitting the value \emph{true}.}
  \label{fig:np:fork}
  \end{figure}

  The core of the fork consists of six roads~$r_1,\dots,r_6$, whereas
  $r_1$, $r_2$, and $r_3$ are vertical line segments and $r_4$, $r_5$
  and $r_6$ are horizontal line segments; see
  Fig.~\ref{fig:np:fork}. We arrange those roads such that $r_1$ and
  $r_2$ have each one junction with~$r_4$ and one junction
  with~$r_5$. Further, $r_3$ has one junction with $r_4$, one
  with~$r_5$ and one with $r_6$. The label length of those roads is
  chosen so that it is exactly the length of the roads. Hence, a
  placed label idenfies all road sections of the roads.

  Further, there are three roads $g_1$, $g_2$, $g_3$ such that $g_1$
  has one junction with~$r_1$, $g_2$ has one junction with $r_2$ and
  $g_3$ has one junction with $r_6$. In all three cases we place the
  junction so that it splits the road in a short road section that is
  shorter than the road's label length and a long road section that
  has exactly the road's label length. We call $g_1$, $g_2$ and $g_3$
  \emph{gates}, because later these roads will be connected to the end
  roads of chains by junctions. To that end those \emph{connecting}
  junctions will be placed on the long road sections of the gates; see
  violet dotted areas in Fig.~\ref{fig:np:fork}.

  The fork gadget has exactly two states, in which 16 road sections
  are identified.  In the first state the labels of $r_1$, $r_2$ and
  $r_3$ are placed; see Fig~\ref{fig:fork:false}.  Hence, the labels
  of $g_1$ and $g_2$ identify only the long road sections of $g_1$ and
  $g_2$, but not the short ones.  The label of $g_3$ idenfies both the
  long and short road section of~$g_3$.  In the second state the
  labels of $r_4$, $r_5$, $r_6$ are placed; see
  Fig~\ref{fig:fork:true}.  Hence, the labels of $g_1$ and $g_2$
  identify the long and short road sections of~$g_1$ and $g_2$, while
  only the long road section of~$g_3$ is identified by a label.  In
  any other configuration fewer road sections are identified by
  labels.  We use the two optimal states of the gadget to represent
  and transmit the values \emph{true} and \emph{false} from one gate
  to the other two gates. More specifically the gates $g_1$ and $g_2$
  are connected with chains that lead to the same literal, while $g_3$
  is connected with a chain that leads to the complementary literal.
	 
  \textit{Variable Gadget.} We define the \emph{variable gadgets}
  simply by connecting chain and fork gadgets into a connected
  component of intersecting roads. This construction already has the
  functionality of a variable gadget: it represents (in a labeling
  identifying the maximum number of road sections) the same truth
  value in all of its branches, synchronized by the fork gadgets, see
  the blue chains and yellow forks in Fig.~\ref{fig:gadgets}a). More
  precisely, we place a sequence of chains linked by fork gadgets
  along the horizontal line on which the variable vertices are placed
  in the drawing $H_\varphi$. Each fork creates a branch of the
  variable gadget either above or below the line. We create as many
  branches above (below) the line as the variable has occurrences in
  positive (negative) clauses in $\varphi$. The first and last chain
  on the line also serve as branches. The synchronization of the
  different branches via the forks is such that either all top
  branches have their road labels pushed away from the line and all
  bottom branches pulled towards the line or vice versa. In the first
  case, we say that the variable is in the state \emph{false} and in
  the latter case that it is in the state \emph{true}. The example in
  Fig.~\ref{fig:gadgets} has two variables set to \emph{true} and
  three variables set to \emph{false}.
	 
  \textit{Clause Gadget.} Finally, we need to create the clause
  gadget, which links three branches of different variables. The core
  of the gadget is a single road that consists of three sub-paths
  meeting in one junction. Each sub-path of that road shares another
  junction with one of the three incoming variable branches. Beyond
  each of these three junctions the final road sections are just long
  enough so that a label can be placed on the section. However, the
  section between the central junction of the clause road and the
  junctions with the literal roads is shorter than the label
  length. The road of the clause gadget has six sections in total and
  we argue that the six sections can only be identified if at least
  one incoming literal evaluates to \emph{true}. Otherwise at most
  five sections can be identified. By construction, each road in the
  chain of a false literal has its label pushed towards the clause,
  i.e., it blocks the junction with the clause road. As long as at
  least one of these three junctions is not blocked, all sections can
  be identified; see Fig.~\ref{fig:gadgets}b). But if all three
  junctions are blocked, then only two of the three inner sections of
  the clause road can be identified and the third one remains
  unlabeled; see Fig.~\ref{fig:gadgets}b).
	 
  \textit{Reduction.} Obviously, the size of the instance $\mathcal M_\varphi$ is
  polynomial in $n$ and $m$. If we have a satisfying variable
  assignment for $\varphi$, we can construct the corresponding road
  labeling and the number of identified road sections is six per
  clause and a fixed constant number $K'$ of sections in the variable
  gadgets, i.e., at least $K=K'+6m$. On the other hand, if we have a road
  labeling with at least $K$ identified sections, each variable
  gadget is in one of its two maximum configurations and each clause
  road has at least one label that covers a junction with a literal
  road, meaning that the corresponding truth value assignment of the
  variables is indeed a satisfying one. This concludes the
  reduction.
\end{proof}

Since~\MaxTotalCovering is an optimization problem, we only present
the~\NP-hardness proof. Still, one can argue that the
corresponding decision problem is \NP-complete by guessing which junctions are covered by which label and then using linear programming for computing the label positions. We omit the technical details. Further, most roads in the reduction are paths, except for the central road in each clause gadget, which is a degree-3 star. In fact, we can strengthen Theorem~\ref{thm:npc} by using a more complex clause gadget instead that uses only paths; see Appendix~\ref{apx:alt-clause}.

\section{An Efficient Algorithm for Tree-Shaped Road Maps}
\label{sec:tree-algorithm}
In this section we assume that the underlying road graph of the road
map is a tree $T=(V,E)$. In Section~\ref{sec:tree:basic-approach} we present a polynomial-time algorithm to  optimally  solve
\begin{wrapfigure}[16]{l}{4.5cm}
  \centering
  \includegraphics[trim=0pt 0pt 0pt .6cm]{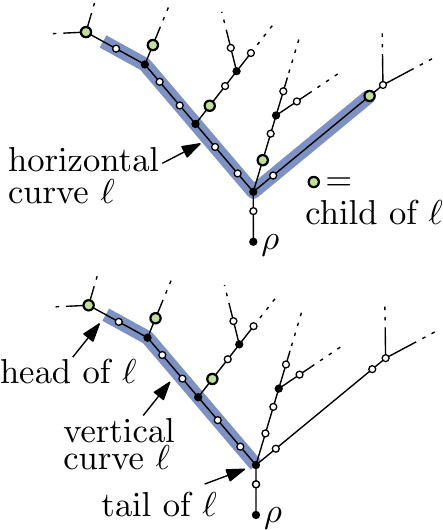}
  \caption{\small Basic definitions.}
  \label{fig:tree:basic-definitions}
\end{wrapfigure}
\MaxTotalCovering for trees; Section~\ref{sec:trees-faster} shows how to improve its running time and space consumption. 
 Our approach uses the 
basic idea that  removing the vertices, whose
embeddings lie in a curve
$c\subseteq \E(T)$, splits the tree into independent parts.
 In 
particular this is true for labels.
We assume that~$T$ is rooted at an arbitrary leaf~$\rho$ and that 
its
edges are
directed away from~$\rho$;~see
Fig.~\ref{fig:tree:basic-definitions}. For two points $p, q \in \E(T)$
we define $\dist(p, q)$ as the length of the shortest curve in~$\E(T)$
that connects~$p$ and~$q$. For two vertices $u$ and $v$ of~$T$ we also
write $\dist(u,v)$ instead of $\dist(\E(u),\E(v))$. For a point~$p\in
E(T)$ we abbreviate the distance~$\dist(p,\rho)$ to the root~$\rho$
by~$\dist_p$. For a curve~$\ell$ in $\E(T)$, we call~$p \in \ell$ the
\emph{lowest point} of~$\ell$ if $\dist_p\leq \dist_q,$ for any~$q\in
\ell$.  As~$T$ is a tree,~$p$ is unique. We distinguish two types of
curves in~$\E(T)$. A curve~$\ell$ is~\emph{vertical} if~$h(\ell)$
or~$t(\ell)$ is the lowest point of~$\ell$; otherwise we call $\ell$
\emph{horizontal} (see
\figurename~\ref{fig:tree:basic-definitions}). Without loss of
generality we assume that the lowest point of each vertical curve
$\ell$ is its tail $t(\ell)$. Since labels are modeled as curves, they
are also either vertical or horizontal.
For a vertex $u\in V$ let $T_u$ denote the subtree rooted at $u$.

\subsection{Basic Approach}\label{sec:tree:basic-approach}

We first determine a finite set of candidate positions for the heads
and tails of labels, and transform $T$ into a tree $T' = (V', E')$ by
subdividing some of $T$'s edges so that it
contains a vertex for every
candidate position. 
To that end we construct for each regular vertex $v \in V$ a chain of
tightly packed vertical labels that starts at $\E(v)$, is directed
towards $\rho$, and ends when either the road ends, or adding the next
label does not increase the number of identified road sections. More
specifically, we place a first vertical label~$\ell_1$ such that
$h(\ell_1) = \E(v)$.  For $i=2, 3, \dots$ we add a new vertical label
$\ell_i$ with $h(\ell_i)$ = $t(\ell_{i-1})$, as long as $h(\ell_i)$
and $t(\ell_i)$ do not lie on the same road section and none of
$\ell_i$'s endpoints lie on a junction edge.
 We use the tails of all
those labels to subdivide the tree $T$. Doing this for all regular
vertices of~$T$ we obtain the tree~$T'$, which we call the
\emph{subdivision tree} of~$T$. The vertices in $V'\setminus V$ are
neither junction vertices nor regular vertices.  Since each chain
consists of $O(n)$ labels the cardinality of~$V'$ is~$O(n^2)$.  We
call an optimal labeling $\OPT$ of $T$ an \emph{canonical labeling} if
for each label $\ell \in \mathcal L'$ there exists a vertex $v$ in~
$T'$ with $\E(v) = h(\ell)$ or $\E(v) = t(\ell)$.  The next lemma
proves that is sufficient to consider canonical labelings.

\newcommand{\thmCanonicalLabeling}{
 For any road graph $T$ that is a tree, there exists a canonical labeling~$\mathcal L$. 
}
\begin{lemma}\label{lem:tree:canonical-labeling}
\thmCanonicalLabeling
\end{lemma}

\begin{proof}
  Let~$\mathcal L$ be an optimal labeling of~$T$. We \emph{push} the
  labels of $\mathcal L$ as far as possible towards the leaves of~$T$
  without changing the identified road sections; see
  Fig.~\ref{fig:tree:packed-labels}. More specifically, starting with
  the labels closest to the leaves, we move each label away from the
  root as far as possible while its head and tail must remain on their
  respective road sections. For a vertical label this direction is
  unique, while for horizontal labels we can choose any of the
  two. Then, for each label its head or tail either coincides with a
  leaf of $T$, with some internal regular vertex, or with the head of
  another label. Consequently, each vertical label belongs to a chain
  of tightly packed vertical labels starting at a regular vertex~$v\in
  V$. Further, the head or tail of each horizontal label coincides
  with the end of a chain of tightly packed vertical labels or a
  regular vertex of~$T$, which proves the claim.  \qed
\end{proof}

We now explain how to construct such a canonical labeling. To that end
we first introduce some notations. For a vertex~$u\in V'$ let~$\OPT(u)$ denote a labeling that identifies
\begin{wrapfigure}[19]{l}{2.8cm}
  \includegraphics[page=2, trim=0pt 0pt 0pt 0.6cm]{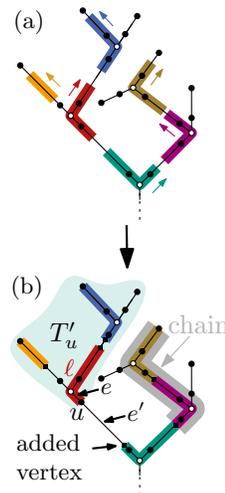}
  \caption{\small Canonical labeling.}
  \label{fig:tree:packed-labels}
\end{wrapfigure}
 a maximum number of road sections
in $T$ only using valid labels in $\E(T'_u)$, where $T'_u$ denotes the subtree of $T'$ rooted at~$u$.
 Note that those labels also may identify the 
incoming road section of $u$, e.g., label~$\ell$ in
Fig.~\ref{fig:tree:packed-labels}b) identifies the edge~$e'$.

Further, the children of a vertex $u\in V'$ are denoted by the
set~$N(u)$; we explicitly exclude the parent of~$u$ from
$N(u)$. Further, consider an arbitrary curve~$\ell$ in $\E(T)$ and
let~$\ell'=\ell\setminus\{t(\ell),h(\ell)\}$. We observe that removing
all vertices of $T'$ contained in $\ell'$ together with their incident
outgoing edges creates several independent subtrees. We call the roots
of these subtrees (except the one containing $\rho$) \emph{children}
of $\ell$ (see Fig.~\ref{fig:tree:basic-definitions}). 
If no vertex of~$T'$ lies in~$\ell'$, the curve is contained in a single edge~$(u,v)\in E'$. In that case~$v$ is the only child
of~$\ell$. We denote the set of all children of $\ell$ as $N(\ell)$. 

For each vertex $u$ in $T'$ we introduce a set $C(u)$ of
\emph{candidates}, which model potential labels with lowest
point~$\E(u)$. If~$u$ is a regular vertex of~$T$ or~$u\in V'\setminus
V$, the set~$C(u)$ contains all vertical labels~$\ell$ with lowest
point~$\E(u)$. If~$u$ is a junction vertex, $C(u)$ contains all
horizontal labels that start or end at a vertex of $T'$ and whose
lowest point is~$\E(u)$. In both cases we assume that~$C(u)$ also
contains the degenerated curve~$\bot_u=\E(u)$, which is the
\emph{dummy label} of~$u$. We set $N(\bot_u)=N(u)$
 and~$\weight(\bot_u)=0$.

For a curve~$\ell$ we define $ \OPT(\ell) = \bigcup_{v\in
  N(\ell)}\OPT(v) \cup \{\ell\}$. Thus, $\OPT(\ell)$ is a labeling comprising $\ell$
and the labels of its children's optimal labelings. We call a label
$\overline \ell\in C(u)$ with $\overline
\ell=\argmax\{\weight(\OPT(\ell)) \mid \ell \in C(u)\}$ an
\emph{optimal candidate} of~$u$. Next, we prove that it is sufficient
to consider optimal candidates to construct a canonical labeling.

\newcommand{\lemOptCandidate}{
  Given a vertex $u$ of $T'$ and an optimal
  labeling~$\OPT(u)$ and let~$\overline \ell$ be an optimal
  candidate of $u$, then it is true that $ \weight(\OPT(u)) =
  \weight(\OPT(\overline \ell)) $.
}
\begin{lemma}\label{lem:tree:basic-approach}
  \lemOptCandidate
\end{lemma}

\begin{proof}
  First note that $\weight(\OPT(u))\geq \weight(\OPT(\overline \ell))$
  because both labelings $\OPT(u)$ and $\OPT(\overline{\ell})$ only
  contain labels that are embedded in $\E(T'_u)$.
  By Lemma~\ref{lem:tree:canonical-labeling} we can assume without
  loss of generality that~$\OPT(u)$ is a canonical labeling.
  Let~$\ell$ be the label of $\OPT(u)$ with $\E(u)$ as the
  lowest point of~$\ell$ (if it exists). 

  If $\ell$ exists, then the vertices in~$N(\ell)$ are roots of
  independent subtrees, which directly yields
  $\weight(\OPT(u))=\weight(\OPT(\ell))$.  By construction of $C(u)$
  we further know that~$\ell$ is contained in~$C(u)$.  Hence, $\ell$
  is an optimal candidate of~$u$, which
  implies~$\weight(\ell)=\weight(\overline \ell)$.
  
  If ~$\ell$ does not exist, then we have
  \[\weight(\OPT(u))=
  \weight(\bigcup_{v\in N(u)}\OPT(v)) \stackrel{(1)}{=}
  \weight(\bigcup_{v\in N(\bot_u)}\OPT(v) \cup \{\bot_u\})
  = \weight(\OPT(\bot_u)).\] Equality $(1)$ follows
  from $N(\bot_u)=N(u)$ and the definition that
  $\bot_u$ does not identify any road section. Since~$\bot_u$ is
  contained in $C(u)$, the dummy label $\bot_u$ is the optimal
  candidate~$\overline \ell$.  \qed
\end{proof}

Algorithm~\ref{algo:basic-approach} first constructs the subdivision tree $T'=(V',E')$ from $T$. Then starting with the leaves of~$T'$ and going to the root~$\rho$ of~$T'$, it computes an optimal candidate~$\overline \ell=$\OPTCAND$(u)$ for each vertex~$u\in V'$ in a bottom-up fashion. By Lemma~\ref{lem:tree:basic-approach} the labeling~$\mathcal L(\overline \ell)$ is an optimal labeling of~$T'_u$. In particular~$\mathcal L(\rho)$ is the optimal labeling of~$T$. 
\begin{algorithm}[b]
  \KwIn{Road graph~$T$, where $T$ is a tree with root $\rho$.} \KwOut{Optimal labeling~$\OPT(\rho)$ of~$T$.}
  $T' \gets $ compute subdivision tree of~$T$\; 
  \lFor{each leaf~$v$ of $T'$}{$\OPT(v)\gets \emptyset$}
  \For{each vertex~$u$ of $T'$ considered in a bottom-up traversal of~$T'$}{
    $\OPT(u) \gets \OPT(\OPTCAND(u))$\;\label{line:opt-candidate} } \Return
  $\OPT(\rho)$
  \caption{Computing an optimal labeling
    of~$T$. }
  \label{algo:basic-approach}
\end{algorithm}

Due to the size of the subdivision tree $T'$ we consider $O(n^2)$ vertices.  Implementing
$\OPTCAND(u)$, which computes an optimal 
candidate~$\overline \ell$ for $u$, naively, creates~$C(u)$ explicitly. We observe that
if $u$ is a junction vertex, $C(u)$ may contain $O(n^2)$
labels; $O(n^2)$ pairs of road sections of different subtrees of $u$
can be connected by horizontal labels. Each label can be constructed
in~$O(n)$ time using a breadth-first search. Thus, for each vertex~$u$
the procedure \texttt{OptCandidate} needs in a naive implementation
$O(n^3)$ time, which yields
$O(n^5)$ running time in total.
Further, we need $O(n^2)$ storage to
store $T'$. Note that we do not need to store $\OPT(u)$ for each
vertex~$u$ of $T'$, but by Lemma~\ref{lem:tree:basic-approach} we can
reconstruct it using $\OPT(\overline \ell)$, where~$\overline \ell$ is
the optimal candidate of~$u$. To that end we store for each vertex of
$T'$ its optimal candidate~$\overline \ell$ and $w(\OPT(\overline
\ell))$.

\begin{theorem}
  For a road map with a tree as underlying road graph,
  \MaxTotalCovering can be solved in~$O(n^5)$ time using $O(n^2)$
  space.
\end{theorem}

In case that all roads are paths, Algorithm~\ref{algo:basic-approach}
runs in~$O(n^4)$ time, because for each $u\in V'$ the set~$C(u)$
contains~$O(n)$ labels. Further, besides the
\emph{primary objective} to identify a maximum number of road sections, Chiri{\'e}~\cite{street-name-placement} also suggested
several additional \emph{secondary objectives}, e.g., labels should
overlap as few junctions as possible.
Our approach allows us to easily incorporate secondary
objectives by changing the weight function $\weight$ appropriately.

\subsection{Improvements on Running Time}\label{sec:trees-faster}

In this part we describe how the running time of
Algorithm~\ref{algo:basic-approach} can be improved to $O(n^3)$ time
by speeding up \texttt{OptCandidate}$(u)$ to~$O(n)$ time.

For an edge~$e=(u,v)\in E\cup E'$ we call a vertical curve
$\ell\subseteq \E(T)$ an \emph{$e$-rooted} curve, if~$t(\ell)=\E(u)$, $h(\ell)$ lies on a road section, and $\length(\E(e)\cap\ell)=\min\{\length(\ell),\length(\E(e))\}$, i.e., $\ell$ emanates 
from~$\E(u)$ passing through~$e$; for example the red label in
Fig.~\ref{fig:tree:packed-labels}b) is an $e$-rooted curve.  An
$e$-rooted curve~$\ell$ is \emph{maximal} if there is no other
$e$-rooted curve~$\ell'$ with $\length(\ell)=\length(\ell')$ and
$\weight(\OPT(\ell'))>\weight(\OPT(\ell))$.  We observe that in any
canonical labeling each vertical label~$\ell$ is a~$(u,v)$-rooted
curve with~$(u,v)\in E'$, and each horizontal label~$\ell$ can be
composed of a $(u,v_1)$-rooted curve~$\ell_1$ and a $(u,v_2)$-rooted
curve~$\ell_2$ with~$(u,v_1),(u,v_2)\in E'$ and~$\E(u)$ is the lowest
point of~$\ell$; see Fig.~\ref{fig:tree:regular-vertex} and
Fig.~\ref{fig:tree:junction-vertex}, respectively.
Further, for a vertical curve~$c$ in $\E(T)$ its
\emph{distance interval $I(c)$} is
$[\dist_{t(c)},\dist_{h(c)}]$.
Since~$T$ is a tree, for every
point $p$ of $c$ we have $\dist_p\in I(c)$.
\begin{wrapfigure}[10]{r}{6cm}
  \centering
  \includegraphics[page=3,trim=0pt 0pt 0pt 0.9cm]{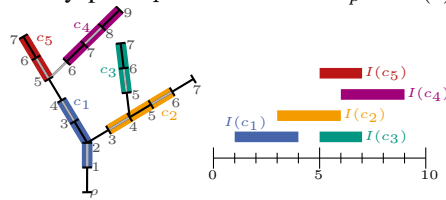}
  \caption{\small Superposing curves, e.g., $c_1$ and
    $c_2$ superpose each other, while $c_1$ and $c_5$ do not.
	The tree is annotated with distance marks.}
  \label{fig:tree:superposition}
\end{wrapfigure}
 Two
vertical curves $c$ and $c'$ 
\emph{superpose} each other if~$I(c)\cap
I(c')\neq \emptyset$; see Fig~\ref{fig:tree:superposition}.

Next, we introduce a data structure that encodes for each
edge~$(u,v)$ of~$T$ all maximal $(u,v)$-rooted curves as~$O(n)$
superposition-free curves in $\E(T_u)$. In particular, each of those curves lies on a
single road section such that all $(u,v)$-rooted curves ending on that
curve are maximal and identify the same number of road sections. We
define this data structure as follows.

\begin{definition}[Linearization]
  Let $e=(u,v)$ be an edge of~$T$.  A tuple
  $(L,\weighting)$ is called a \emph{linearization} of~$e$, if $L$ is
  a set of superposition-free curves and $\weighting\colon
  L\to\mathbb{R}$ such that
  \begin{compactenum}[(1)]
  \item for each curve $c\in L$ there is a road section~$e'$ in $T_u$
    with $c\subseteq \E(e')$\label{lin:cond1},
  \item for each~$e$-rooted curve~$\ell$ there is a curve~$c\in L$
    with $\length(\ell)+\dist_u \in I(c)$\label{lin:cond2},
  \item for each point $p$ of each curve~$c\in L$ there is a maximal
    $e$-rooted curve~$\ell$ with $h(\ell)=p$ and
    $\weighting(c)=\weight(\OPT(\ell))$.\label{lin:cond3}
  \end{compactenum}
\end{definition}

Assume that we apply Algorithm~\ref{algo:basic-approach} on $T'$ and that
we currently consider the vertex $u$ of $T'$. Hence, we can assume
that for each vertex $v\neq u$ of $T'_u$ its optimal candidate and
$\weight(\OPT(v))$ is given.  We first explain how to speed up~\texttt{OptCandidate}
using linearizations. Afterwards, we present the construction of
linearizations.

\subsubsection{Application of linearizations.}  Here we assume
that the linearizations are given for the
edges of~$T$. Concerning the type of $u$ we
describe how to compute its optimal candidate.

\textit{Case 1, $u$ is regular.} If~$u$ is a leaf, the set $C(u)$
contains only~$\bot_u$.  Hence, assume that~$u$ has one outgoing
edge~$e=(u,v)\in E'$, which belongs to a road~$R$. Let $P$ be the
longest path of vertices in~$T'_u$ that starts at~$u$ and does not
contain any junction vertex.  Note that the path must be
unique. Further, by construction of~$T'$ the last vertex~$w$ of $P$
must be a regular vertex in~$V$, but not in $V'\setminus V$. We
consider two cases; see Fig~\ref{fig:tree:regular-vertex}.
\begin{wrapfigure}[12]{l}{3.0cm}
  \centering
  \includegraphics[page=4,trim=0 0 0 0.4cm]{fig/tree1.pdf}
  \caption{\small Case 1}
  \label{fig:tree:regular-vertex}
\end{wrapfigure}
\indent If $\dist(u,w)\geq \lw(R)$, the optimal candidate is either $\bot_u$ or
the $e$-rooted curve~$\ell$ of length~$\lw(R)$ that ends on
$\E(P)$. By assumption and due to
$\weight(\OPT(\bot_u))=\weight(\OPT(v))$, we decide in $O(1)$ time
whether $\weight(\OPT(\bot_u))\geq \weight(\OPT(\ell))$, obtaining the optimal candidate.

If $\dist(u,w) < \lw(R)$, the optimal candidate is either~$\bot_u$ or
goes through a junction. Since~$w$ is regular, it has only one
outgoing edge~$e'=(w,x)$.  Further, by the choice of~$P$ the edge~$e'$
is a junction edge in $T$; therefore the linearization
$(L,\weighting)$ of $e'$ is given. 
 In linear time we search for the
curve $c \in L$ such that there is an $e$-rooted curve~$\ell$ of
length~$\lw(R)$ with its head on $c$. To that end we consider for each
curve~$c \in L$ its distance interval $I(c)$ and check whether there
is $t\in I(c)$ with $t-\dist_u=\lw(R)$. Note that using a binary search tree 
for finding~$c$ speeds this procedure up to $O(\log n)$ time,
however, this does not asymptotically improve the total running time.
  The $e$-rooted curve~$\ell$ then can be easily constructed in
$O(n)$ time by walking from~$c$ to~$u$ in~$\E(T)$.  

If such a
curve~$c$ exist, by definition of a linearization the optimal
candidate is either $\bot_u$ or $\ell$, which we can decide in $O(1)$
time by checking $\weight(\OPT(\bot_u))\geq \weight(\OPT(\ell))$. 
\begin{wrapfigure}[8]{l}{3.0cm}
  \centering
  \includegraphics[page=5,trim=0 0 0 0.6cm]{fig/tree1.pdf}
  \caption{\small Case 2}
  \label{fig:tree:junction-vertex}
\end{wrapfigure}
Note
that we have $\weight(\OPT(\bot_u))=\weight(\OPT(v))$ and
$\weight(\OPT(\ell))=\overline \weight(c)$.  
If $c$ does not exist,
again by definition of a linearization there is no vertical
label~$\ell \in C(u)$ and $\bot_u$ is the optimal candidate.

\textit{Case 2, $u$ is a junction vertex.} The set $C(u)$ contains
horizontal labels.
Let~$\ell$ be such a label and let $e_1=(u,v_1)$ and $e_2=(u,v_2)$ be
two junction edges in $E$ covered by~$\ell$; see Fig.~\ref{fig:tree:junction-vertex}. 
Then
there is an $e_1$-rooted curve $\ell_1$ and an $e_2$-rooted curve
$\ell_2$ whose
composition is~$\ell$.
 Further,
we have $\weight(\OPT(\ell))=\weight(\OPT(\ell_1)\cup
\OPT(\ell_2))+\sum_{v\in N(u)\setminus\{v_1,v_2\}}\weight(\OPT(v))
$. We use this as follows.

Let $e_1$ and $e_2$ be two outgoing edges of~$u$ that belong to the
same road~$R$, and let $(L_1,\weighting_1)$ and $(L_2,\weighting_2)$
be the linearizations of $e_1$ and $e_2$, respectively.
We define for
$e_1$ and $e_2$ and their linearizations the operation $\ml(L_1,L_2)$
that finds an optimal candidate of $u$ restricted to labels
identifying $e_1$ and $e_2$.

For~$i=1,2$ let~ $d_i=\max\{\dist_u\mid u
\text{ is vertex of }T_{v_i}\}$ and
let~$f_u(t)=\dist_u-(t-\dist_u)=2\dist_u-t$ be the function that
``mirrors'' the point~$t\in \mathbb{R}^2$
at~$\dist_u$. 
\begin{wrapfigure}[17]{l}{6.2cm}
  \centering
  \includegraphics[page=6,trim=0 0 0 0.5cm]{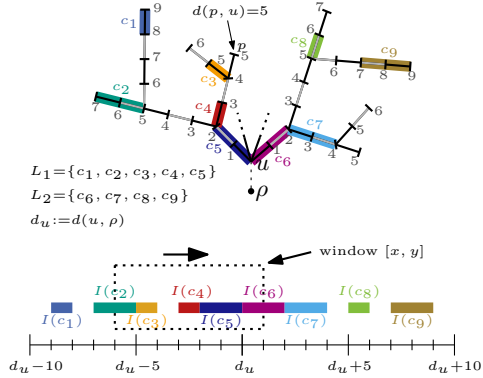}
  \caption{\small Constructing the optimal candidate of  $u$
    based on the linearizations $(L_1,\weighting_1)$
    and~$(L_2,\weighting_2)$. The tree is annotated with distance
    marks.}
  \label{fig:tree:merge-linearizations}
\end{wrapfigure}
Applying~$f_u(t)$ on the boundaries of the distance
intervals of the curves in $L_1$, we first mirror these intervals such
that they are contained in the interval~$[2\dist_u-d_1,\dist_u]$; see
Fig.~\ref{fig:tree:merge-linearizations}. Thus, the curves in~$L_1 \cup L_{2}$
are mutually superposition-free such that their distance
intervals lie in~$J=[2\dist_u-d_{1},d_{2 }]$.

\noindent We call an interval $[x,y] \subseteq J$ a \emph{window}, if
it has length $\lw(R)$, $\dist_u\in [x,y]$ and there are curves
$c_1\in L_1$ and $c_2\in L_2$ with $x\in I(c_1)$ and $y\in I(c_2)$;
see Fig.~\ref{fig:tree:merge-linearizations}. By the definition of a
linearization there is a maximal $e_1$-rooted curve $\ell_1$ ending on
$c_1$ and a maximal $e_2$-rooted curve $\ell_2$ ending on $c_2$ such
that $\length(\ell_1)+\length(\ell_2)=\lw(R)$.  Consequently, the
composition of $\ell_1$ and $\ell_2$ forms a horizontal label~$\ell$
with $\weight(\OPT(\ell))=\weight(\OPT(\ell_1)\cup
\OPT(\ell_2))+\sum_{v\in N(u)\setminus\{v_1,v_2\}}\OPT(v)$; we call
$\weight(\OPT(\ell))$ the \emph{value} of the window.  Using a simple
sweep from left to right we compute for the distance interval~$I(c)$
of each curve $c \in L_1\cup L_2$ a window~$[x,y]$ that starts or ends
in $I(c)$ (if such a window exists). The result of $\ml(L_1,L_2)$ is
then the label $\ell$ of the window with maximum value.
For each pair $e_1$ and $e_2$ of outgoing edges we apply
$\ml(L_1,L_2)$ computing a label $\ell$. By construction either the
label~$\ell$ with maximum $\weight(\ell)$ or $\bot_u$ is the optimal
candidate for~$u$, which we can check in $O(1)$ time.
Later on we prove that we  consider only linearizations of linear size. Since each vertex 
of~$T'$ has constant degree, we obtain the next lemma.

\begin{lemma}\label{lem:tree:apply-lin}
  For each $u\in V'$ the optimal candidate can be found in
  $O(n)$ time.
\end{lemma}

\subsubsection{Construction of linearizations.}

\begin{figure}
  \centering
  \begin{minipage}[t]{0.35\textwidth}
    \centering
    \includegraphics[page=7]{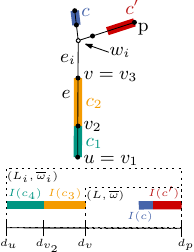}
    \caption{\small 1st Step: For each edge~$e_i$ extend its linearization~$(L,\weighting)$ to a 
    linearization~$(L_i,\weighting_i)$ of~$T_i$.}
    \label{fig:tree:construction-step1}
  \end{minipage}
  \hfill
  \begin{minipage}[t]{0.6\textwidth}
    \centering
    \includegraphics[page=9]{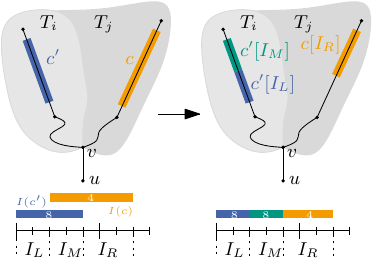}
    \caption{\small 2nd Step: Merging the linearizations of the trees $T_i$ and $T_j$.}
    \label{fig:tree:construction-step2}
  \end{minipage}
\end{figure}

We now show how to recursively construct a linearization for an edge
$e=(u,v)$ of~$T$.  To that end we assume that we are given the
subdivision tree $T'$ of $T$ and the linearizations for the outgoing
edges~$e_1=(v,w_1),\dots,e_k=(v,w_k)$ of~$v$ that belong to the same
road~$R$ as~$e$. Further, we can assume that we have computed the
weight $\weight(\OPT(w))$ for all vertices~$w$ in $T'_u$ excluding
$u$. In case that two vertices of those vertices share the same
position in $\E(T'_u)$ we remove that one with less weight.  Let $T_i$
be the tree induced by the edges $e$, $e_i$ and the edges of the
subtree rooted at~$w_i$.  As a first step we compute for each
linearization $(L,\weighting)$ of each edge~$e_i$ a linearization
$(L_i,\weighting_i)$ for $e$ restricted to tree~$T_i$, i.e.,
conceptually, we assume that $T_u$ only consists of $T_i$'s edges. 

\noindent If $e$ is a junction edge we set $L_i\gets L$ and weight each curve~$c\in
L_i$ as follows.
\[
\weighting_i(c)\gets \weighting(c)+\sum\limits_{\mathclap{w\in
    N(v)\setminus\{w_i\}}}\weight(\OPT(w))
\]

Otherwise, if $e$ is a road section, let $v_1,\dots,v_l$ be the
vertices of the subdivision tree~$T'$ that lie on~$e$, i.e.,
$\E(v_j)\in \E(e)$ for all $1\leq j\leq l$; see
Fig.~\ref{fig:tree:construction-step1}. We assume that $\dist(v_1)<
\ldots< \dist(v_l)$, which in particular yields $v_1=u$ and
$v_l=v$. Let $c_1$ be the curve~$\E((v_1,v_2))$ and for $2\leq j<l$
let~$c_j$ be the curve $\E((v_j,v_{j+1}))\setminus \E(v_j)$.
Hence, we have $\bigcup_{j=1}^l c_j=\E(e)$ and $c_j\cap
c_{j'}=\emptyset$ for $1\leq j < j' < l$.  We set
\[
 L_i\gets L\cup \bigcup_{j=1}^{l-1}\{c_j\}
\]
We weight each curve $c\in L_i$ as follows.
If $c$ is contained in $L$, we set 
\[
\weighting_i(c)\gets \weighting(c)+1
\]
Otherwise, $c$ is a sub-curve of $\E(e)$ and there exists a $j$ with $c=c_j$. We set 
\[
  \weighting_i(c)\gets\weight(\mathcal L(v_{j+1})\cup \{\ell_c\}), 
\]
where $\ell_c\subseteq \E(e)$ is an $e$-rooted curve that starts at $\E(u)$ and
ends on $c$. The next lemma
shows that this transformation yields a linearization as desired.
\begin{lemma}\label{lem:tree:construct1}
  For each outgoing edge $e_i$ with linearization $(L,\weighting)$ the
  tuple $(L_i,\weighting_i)$ is a linearization of $e$ restricted to
  the tree $T_i$.
\end{lemma}

\begin{proof}
  We use the same notation as used above.

  First of all, the set $L_i$ contains only curves that do not
  superpose each other: Since $L_i$ contains only curves that do not
  superpose each other, the only curves that could superpose another
  curve in $L$ are contained in $L_i\setminus L$. Since $L_i\setminus
  L$ is empty for a junction edge, we can assume that~$e$ is a road
  section. By construction those curves in $L_i\setminus L$
  partition~$\E(e)$ without intersecting each other. Further,
  by assumption no two road sections share a common vertex and since
  all curves of~$L$ are contained in $\E(T_v)$, the curves in
  $L_i\setminus L$ cannot superpose any curve in~$L$.

  We now prove that $L_i$ satisfies the three conditions of a
  linearization.  First assume that $e$ is a road section.

  \emph{Condition~(\ref{lin:cond1})}. Since $L$ is a linearization,
  each curve of $L$ must be a sub-curve of a road section. Further,
  the curves $L_i\setminus L$ are sub-curves of the road section~$e$.

  \emph{Condition~(\ref{lin:cond2})}. First consider an $e$-rooted
  curve~$\ell$ that either ends on $e_i$ or on an edge
  of~$T_{w_i}$. Recall that $h(\ell)$ must lie on a road section. Then
  there is an~$e_i$-rooted curve~$\ell'$ with~$\ell' \subseteq \ell$
  and~$h(\ell)=h(\ell')$. Hence, there is a curve~$c\in L$ with
  $\length(\ell')+\dist_v\in I(c)$. Since~$\ell'$ is a sub-curve
  of~$\ell$, we also have $\length(\ell)+\dist_u\in I(c)$. Now,
  consider an $e$-rooted curve~$\ell$ that ends on $e$, then~obviously
  by construction there is a curve~$c\in L_i\setminus L$ with
  $\length(\ell)+\dist_u\in I(c)$.

  \emph{Condition~(\ref{lin:cond3})}. First consider an arbitrary
  curve~$c\in L_i\setminus L$ and let~$\ell$ be any $e$-rooted curve 
  that ends on $c$. Further, let $v_1,\dots,v_l$ be the
  vertices of the subdivision tree~$T'$ that lie on~$e$ as defined
  above. By construction there is an edge~$(v_j,v_{j+1})$ with $1\leq
  j < l$ and $c\subseteq \E(v_j,v_{j+1})$. It holds
  \[
  \weight(\OPT(\ell))=\weight(\OPT(v_{j+1})\cup\{\ell\})=\weighting_i(c)
  \]
  Obviously, $\ell$ must be maximal, because there is no other point
  in $\E(T_i)$ having the same distance to~$\rho$ as $h(\ell)$ has.

  Finally, consider a curve~$c\in L$ and let~$\ell$ be any
  $e$-rooted curve that ends on~$c$. As~$L$ is a linearization
  of~$e_i$, for each point~$p$ on $c$ there must be an $e_i$-rooted
  curve~$\ell'$ with~$h(\ell')\in c$. We choose~$\ell'$ such
  that~$h(\ell')=h(\ell)$.  Since~$\ell'$ is a maximal $e_i$-rooted
  curve, the curve $\ell$ must be a maximal $e$-rooted curve. Further,
  $\ell$ identifies one road section more than~$\ell'$. Hence, we
  obtain
  \[\weight(\OPT(\ell))=\weight(\OPT(\ell'))+1= \weighting(c)+1
=\weighting_i(c)\]

Now consider the case that~$e$ is a junction
edge. \emph{Condition~(\ref{lin:cond1})} and
\emph{Condition~(\ref{lin:cond2})} follow by the same arguments as
stated above with the simplification that $L_i=L$.
  
\emph{Condition~(\ref{lin:cond3})}. Let~$c$ be a curve in $L_i$ and
let $\ell$ be any $e$-rooted curve that ends on~$c$. Further,
let~$\ell'$ be the $e_i$-rooted sub-curve of~$\ell$ that starts at
$\E(v)$ and ends at $h(\ell)$; by definition of $L$ such a curve
exists. It holds
 \[
    \weight(\OPT(\ell))=\weight(\OPT(\ell'))+\sum\limits_{\mathclap{w\in
    N(v)\setminus\{w_i\}}}\weight(\OPT(w))=\weighting(c)+\sum\limits_{\mathclap{w\in
    N(v)\setminus\{w_i\}}}\weight(\OPT(w))=\weighting_i(c)
 \]
 Since~$\ell'$ is a maximal~$e_i$-rooted curve, it directly follows that
 $\ell$ is a maximal $e$-rooted curve with respect to~$T_i$.
\qed
\end{proof}

In the next step we define an operation $\oplus$ by means of which two
linearizations $(L_i,\weighting_i)$ and $(L_j,\weighting_j)$ can be
combined to one linearization $(L_i,\weighting_i) \oplus
(L_j,\weighting_j)$ of $e$ that is restricted to the subtree~$T_{i,j}$
induced by the edges of $T_i$ and $T_j$. Consequently,
$\bigoplus_{i=1}^k (L_i,\weighting_i)$ is the linearization of $e$
without any restrictions.

\begin{figure}[h]
  \centering
  \includegraphics[page=8]{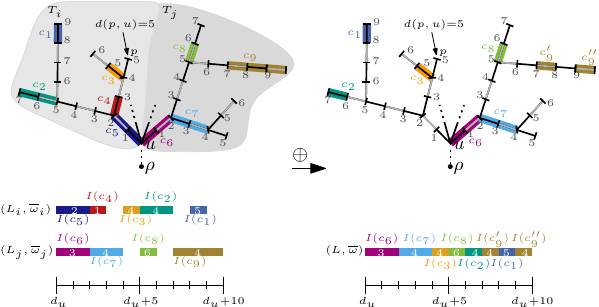}
  \caption{Illustration of merging two linearizations~$(L_i,\weighting_i)$ and
   $(L_j,\weighting_j)$ into one linearization $(L_1,\weighting_i)$. The trees are
   annotated with distance marks.}
  \label{fig:tree:merge}
\end{figure}

We define $(L,\weighting)=(L_i,\weighting_i) \oplus
(L_j,\weighting_j)$ as follows; for illustration see also
Fig.~\ref{fig:tree:merge}.  Let $c_1,\dots,c_\ell$ be
the curves of $L_{i}\cup L_{j}$ such that for any two curves $c_s$,
$c_t$ with $s<t$ the left endpoint of~$I(c_s)$ lies to the left of
the left endpoint of $I(c_t)$; ties are broken arbitrarily. We
successively add the curves to~$L$ in the given order enforcing that
the curves in $L$ remain superposition-free. Let $c$ be the next curve
to be added to $L$.

Without loss of generality, let~$c\in L_{i}$. The opposite case can be
handled analogously. In case that there is no curve superposing $c$,
we add $c$ to $L$ and set $\weighting(c)=\weighting_{i}(c)$.  If $c$
superposes a curve in $L$, due the order of insertion, there can only
be one curve $c'$ in $L$ that superposes $c$.
First we remove $c'$ from $L$. Let $I_M$ be the interval describing
the set $I(c)\cap I(c')$, and let $I_L$ and $I_R$ be the intervals
describing the set~$I(c)\cup I(c') \setminus (I(c) \cap I(c'))$ such
that $I_L$ lies to the left of~$I_M$ and~$I_R$ lies to the right of
$I_M$; see Fig.~\ref{fig:tree:construction-step2}.

We now define three curves $c_L$, $c_M$ and $c_{R}$ with $I(c_L)=I_L$,
$I(c_M)=I_M$ and $I(c_R)=I_R$ such that each of these three curves is
a sub-curve of either~$c$ or~$c'$. To that end let $c[I]$ denote the
sub-curve of $c$ whose distance interval is $I$.  We define the curve
$c_R$ with weight $\weighting(c_R)$ as
\[
(c_R,\weighting(c_R)) =
\begin{cases}
  (c[I_R],\weighting_i(c)), & \text{if }I_R\subseteq I(c)\\
  (c'[I_R],\weighting(c')), & \text{if } I_R\subseteq I(c')\\
\end{cases}
\]
The curve $c_L$ and its weight $\weighting(c_L)$ is defined
analogously. The curve~$c_M$ and its weight $\weighting(c_M)$ is
\[
(c_M,\weighting(c_M)) =
\begin{cases}
  (c[I_M],\weighting_i(c)), & \text{if }\weighting_{i}(c)\geq \weighting(c')\\
  (c'[I_M],\weighting(c')), & \text{if }\weighting_{i}(c) < \weighting(c')\\
\end{cases}
\]
The next lemma proves that $(L_i,\weighting_i) \oplus
(L_j,\weighting_j)$ is a restricted linearization.

\begin{lemma}\label{lem:tree:construct2}
  Let $(L_i,\weighting_i)$ and $(L_j,\weighting_j)$ be two
  linearizations of $e=(u,v)$ that are restricted to the trees $T_i$
  and $T_j$, respectively.
  Then $(L,\weighting)=(L_i,\weighting_i) \oplus (L_j,\weighting_j)$
  is a linearization of $e$ restricted to~$T_{i,j}$.  The operation
  needs $O(|L_i|+|L_j|)$ time.
\end{lemma}

\begin{proof}
  First of all, the set $L$ contains only curves that are pairwise
  free from any superpositions. This directly follows from the
  construction that curves~$c$ and $c'$ superposing each other are
  replaced by three superposition-free curves $c_L$, $c_M$ and
  $c_R$. Due to $I(c_L)\cup I(c_M) \cup I(c_R)=I(c)\cup I(c')$ the
  first and second condition of a linearization is satisfied.

  We finally prove that Condition~(\ref{lin:cond3}) of a linearization is satisfied by
  doing an induction over the curves inserted to~$L$. Let~$L^k$ be~$L$
  after the $k$-th insertion step. Since $L^0$ is empty,
  the condition obviously holds for $L^0$. So
  assume that we insert $c$ to $L^k$ obtaining the set $L^{k+1}$.
  Without loss of generality assume
  that~$c\in L_i$. If $c$ does not superpose any curve in $L^k$, the
  condition directly follows from the definition of~$c$.  So assume
  that~$c'\in L^k$ superposes $c$. Since~$c\in L_i$, the curve~$c'$ is
  contained in~$\E(T_j)$. We remove $c'$ from~$L^k$ and insert the
  curves $c_R$, $c_M$ and $c_L$ as defined above. We prove that all
  three curves satisfy Condition~(\ref{lin:cond3}).

  Consider in the following the subtree~$T_{i,j}$ of $T_u$ restricted
  to the edges of~$T_i$ and~$T_j$. We set $c_R= c[I_R]$ and set
  $\weighting(c_R)=\weighting_{i}(c)$, if $I_R\subseteq I(c)$.  In
  that case there is no $e$-rooted curve~$\ell\subseteq \E(T_j)$
  with~$\length(\ell)+\dist_u \in I_R$, i.e., either there is no
  curve~$\ell$ in $\E(T_j)$ with $t(\ell)=\E(u)$ and
  $\length(\ell)+\dist_u \in I_R$, or any curve in $\E(T_j)$ with
  $t(\ell)=\E(u)$ and $\length(\ell)+\dist_u \in I_R$ ends on a
  junction edge.  Consequently, any $e$-rooted curve $\ell$
  with~$\length(\ell)+\dist_u \in I_R$ and in particular any maximal
  $e$-rooted curve~$\ell$ with~$\length(\ell)+\dist_u \in I_R$ lies in
  $\E(T_i)$. Thus, the curve~$c_R$ satisfies
  Condition~(\ref{lin:cond3}). For the case $I_R\subseteq I(c')$ and
  the curve $c_L$ we can argue analogously.

  So consider the curve~$c_M$. Without loss of generality we assume
  that~$\weighting_i(c)\geq \weighting(c')$. The opposite case can be
  handled analogously. For any maximal $e$-rooted curve~$\ell$ in
  $\E(T_j)$ with $\length(\ell)+\dist_u \in I_M$ it must be true that
  $\weight(\OPT(\ell))\leq \weighting(c_M)$.  Further,
  since~$c_M\subseteq c$ and $c$ satisfies condition~(\ref{lin:cond3})
  with respect to~$T_i$, $c_M$ satisfies the
  condition~(\ref{lin:cond3}) with respect to~$T_{i,j}$. \qed
\end{proof}

Lemma~\ref{lem:tree:construct1} and Lemma~\ref{lem:tree:construct2}
yield that $\bigoplus_{i=1}^k (L_i,\weight_i)$ is the linearization of
$e$ without any restrictions. Computing it needs
$O(\sum_{i=1}^k|L_i|)$ time.

Note that when computing optimal candidates (see \emph{Application
  of linearizations}) we are only interested in $e$-rooted
curves~$\ell$ that have length at most $\lw(R)$, where $R$ is the road
of~$e$. Hence, when constructing $(L_i,\weighting_i)$ for an
edge~$e_i$ in the first step, we discard any curve~$c$ of~$L_i$ that
does not allow an $e$-rooted curve that both ends on~$c$ and has
length at most~$\lw(R)$; the curve $c$ is not necessary for our
purposes.  Hence, we conceptually restrict $T_i$ to the edges that are
reachable from~$u$ by one label length. It is not hard to see
that~$T'$ restricted to $\E(T_i)$ contains only $O(n)$ vertices,
because each vertex of~$V'\setminus V$ is induced by a chain of
tightly packed vertical labels, whereas each label has length
$\lw(R)$. Hence, $T'$ restricted to $\E(T_i)$ contains for each such
chain at most one vertex of $V'\setminus V$. Further, the endpoints
of the curves in~$L_i$ are induced by the vertices of $T'$. Hence, by
discarding the unnecessary curves of~$L_i$ the set $L_i$ has size
$O(n)$. Altogether, by Lemma~\ref{lem:tree:construct2} and due to the
constant degree of each vertex we can construct $\bigoplus_{i=1}^k
(L_i,\weight_i)$ in $O(\sum_{i=1}^k n)=O(n)$ time.

When constructing $\OPT(u)$ for $u$ as described in
Algorithm~\ref{algo:basic-approach}, we first build the
linearization~$L_{e}$ of each of $u$'s outgoing edges. By
Lemma~\ref{lem:tree:apply-lin} we can find in $O(n)$ time the optimal
candidate of $u$. Then, due to the previous reasoning, the
linearization of an edge of $T$ and the optimal candidate of a
vertex~$u$ can be constructed in~$O(n)$ time. Altogether we obtain the
following result.

\begin{proposition}
  \label{prop:tree}
  For a road map~$\mathcal M$ with a tree~$T$ as underlying road graph,
  \MaxTotalCovering can be solved in~$O(n^3)$ time.
\end{proposition}

\subsection{Improvements on Storage Consumption}
Since~$T'$ contains $O(n^2)$ vertices, the algorithm needs~$O(n^2)$
space. This can be improved to~$O(n)$ space. To that end~$T'$ is
constructed \emph{on the fly} while executing
Algorithm~\ref{algo:basic-approach}. Parts of~$T'$ that become
unnecessary are discarded. We prove that it is sufficient to store $O(n)$ vertices of~$T'$ at any
time such that the optimal labeling can still be constructed.

When constructing the optimal labeling of~$T$, we build for each
edge~$(u,v)$ of $T$ its linearization based on the linearization of
the outgoing edges of $v$.
\begin{wrapfigure}[13]{l}{3.7cm}
  \centering 
  \includegraphics[page=10]{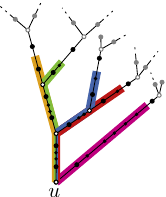}
  \caption{\small Vertices not reachable from $u$ are marked gray.}
  \label{fig:tree:reachable}
\end{wrapfigure}
  Afterwards we discard the linearizations
of those outgoing edges.  Since each vertex has constant degree,
considering the vertices of $T'$ in an appropriate order, it is
sufficient to maintain a constant number of linearizations at any
time.

Hence, because each linearization has size $O(n)$, we need $O(n)$
space for storing the required linearizations in total. However, we
store for each vertex~$u$ of~$T'$ the weight~$\weight(\OPT(u))$ and
its optimal candidate.  As~$T'$ has size $O(n^2)$ the space
consumption is $O(n^2)$. In the following we improve that bound
to~$O(n)$ space.

\noindent We call a vertex~$v\in V'$ \emph{reachable} from a vertex~$u\in V'$,
if there is a curve~$\ell\subseteq \E(T'_u)$ that starts at~$\E(u)$
and that is contained in the embedding of a road~$R$ with
$\lw(R)\geq \length(\ell)$ such that~$\E(v)\in \ell$ or $v\in
N(\ell)$, where~$\length(\ell)$ denotes the length of~$\ell$; see
Fig.~\ref{fig:tree:reachable}. The set~$\R_u$ contains all vertices
of~$T'_u$ that are reachable from~$u$. The next lemma shows
that~$\R_u$ has linear size.

\newpage
\begin{lemma}\label{lem:tree:reachable}
  For any vertex~$u$ of $T'$ the set $\R_u$ has size~$O(n)$.
\end{lemma}

\begin{proof}
  Recall how $T'$ is constructed: For each vertex $v \in V$ we
  construct a chain~$C$ of tightly packed vertical valid labels,
  which starts at $\E(v)$, is directed towards $\rho$, and ends when
  either the road ends, or adding the next label does not increase
  the number of identified road sections.
  Each label of such a chain $C$ induces one vertex of
  $T'$. Hence,~$C$ induces a set~$V_C$ of vertices in $T'$. We show
  that for each chain $C$ the set $V_C\cap \R_u$ contains at most
  two vertices. As we construct $n$ chains in order to build $T'$
  the claim follows.

  For the sake of contradiction assume that there is a chain $C$ and
  a vertex~$u$ in $T'$ such that~$V_C\cap \R_u$ contains more than
  two vertices. Without loss of generality we assume that $V_C\cap
  \R_u$ contains three vertices, which we denote by $v_1$, $v_2$ and
  $v_3$. We further assume that $\dist_{v_1} < \dist_{v_2} <
  \dist_{v_3}$.  By construction all labels in~$C$ lie in the
  embedding of the same road~$R_C$, and $\dist(v_1,v_2)\geq
  \lw(R_C)$ and $\dist(v_2,v_3)\geq \lw(R_C)$. By definition of $C$
  there is a vertical curve~$\ell\in \E(T'_u)$ that starts at
  $\E(u)$ and contains $v_1$, $v_2$ and $v_3$.  Let~$e$ be the
  outgoing edge of~$u$ in $T'$ whose embedding is covered by $\ell$ and
  consider the sub-curve~$\ell' \subseteq \ell$ with
  length~$\lw(R_C)$ that starts at~$u$. By definition of~$\R_u$, we
  know for each~$v_i$ with $1\leq i \leq 3$ that either its
  embbeding is contained in~$\ell'$ or $v_i\in N(\ell')$. From the
  definition of~$N(\ell')$ and the fact that all three vertices lie
  on~$\ell$, it directly follows that only~$v_3$ may be contained
  in~$N(\ell')$. Hence, $\E(v_1),\E(v_2)\in \ell'$. Further,
  because~$v_2\not \in N(\ell')$, we have $\E(v_2)\neq h(\ell')$,
  which implies $\dist(v_1,v_2)<\lw(R)$ and contradicts
  $\dist(v_1,v_2)\geq\lw(R)$.\qed
\end{proof}

Assume that we apply Algorithm~\ref{algo:basic-approach} considering
the vertex~$u$.  When constructing $u$'s optimal candidate, by
Lemma~\ref{lem:tree:reachable} it is sufficient to consider the
vertices of~$T'_u$ that lie in~$\R_u$.  On that account we discard
all vertices of $T'_u$ that lie in $V'\setminus V$, but not
in~$\R_u$.
\begin{wrapfigure}{l}{3cm}
  \centering \vspace{-0.6cm}
  \includegraphics[page=11]{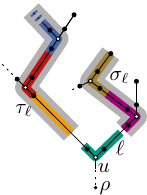}
  \caption{\small Chains of label~$\ell$.}
  \label{fig:tree:reconstructing}
  \vspace{-0.4cm}
\end{wrapfigure}
  Further, we compute the vertices of $V'\setminus V$ that
subdivide the incoming edge~$(t,u)\in E$ \emph{on demand}, i.e., we
compute them, when constructing the optimal candidate of~$t$. Hence,
we have linear space consumption.

\noindent However, when discarding vertices of $T'$, we lose the
possibility of reconstructing the labeling. We therefore annotate
each vertex~$u\in V$ of the original tree~$T$ with further
information. To that end consider a canonical labeling~$\mathcal L$
of~$T$.  Let~$\ell$ be a horizontal label of~$\mathcal L$ and
let~$e$ be the edge of~$T$ on which~$\ell$'s head is
located. Either, no other label of~$\mathcal L$ ends on $e$, or
another label~$\ell'$ ends on $e$ that belongs to a
chain~$\sigma_\ell$ of tightly packed vertical labels. Analogously,
we can define the chain~$\tau_\ell$ with respect to edge~$e'$ on
which $\ell's$ tail is located. On that account we store for a
junction vertex~$u\in V$ not only its optimal candidate~$\ell \in
C(u)$, but also the two chains~$\sigma_\ell$ and $\tau_\ell$, if
they exist. Note that such a chain of tightly packed vertical labels
is uniquely defined by its start and endpoint, which implies that
$O(1)$ space is sufficient to store both chains. Using a
breadth-first search we can easily reconstruct those chains in
linear time. For a regular vertex~$u\in V$ we analogously store
$\sigma_\ell$ of its optimal candidate~$\ell\in C(u)$, if it
exists. Since~$\ell$ is vertical, we do not need to consider its
tail. For the special case that~$\ell=\bot_u$, we define
that~$\sigma_\ell$ is the chain of tightly packed vertical labels
that ends on the only outgoing edge~$e$ of~$u$. Summarizing, the
additional information together with the optimal candidates stored
at the vertices~$u\in V$ of the original tree are sufficient to
reconstruct the labeling of~$T$.
Together with Proposition~\ref{prop:tree} we obtain the following result.

\begin{theorem}
  \label{thm:tree}
  For a road map~$\mathcal M$ with a tree~$T$ as underlying road graph,
  \MaxTotalCovering can be solved in~$O(n^3)$ time using~$O(n)$ space.
\end{theorem}
 
\section{On the Usefulness of Labeling Tree-Shaped Road Networks}\label{sec:appendix:experiments}
Although the underlying road graphs of real-world road maps are rarely
trees, our algorithm for labeling trees is still of practical interest as
we show in first initial experiments. The obtained data shall give the reader
evidence of the practicability and relevance of our algorithm, but
they are not yet a complete experimental study. For a companion
paper we are working on a detailed evaluation of our approach and are
investigating several practical heuristics that are based on the tree labeling algorithm.

To evaluate the usefulness of our algorithm we considered the road networks
of several large cities. We extracted the road graphs from the data provided by
OpenStreetMap\footnote{\url{openstreetmap.org}} and drew them
mimicking the style used on \url{openstreetmap.org} as standard. In
particular, we adapted the zoom level $17$, which maps $50m$ to $65$
pixels.

On each road graph we first applied a simple preprocessing strategy
removing and cutting road sections that can be labeled trivially
without violating any optimal solution. In particular we applied the
following rules.
\begin{compactenum}
  \item Remove any road that contains exactly one road section.
  \item Remove any road section that is sufficiently long to completely contain a label and whose adjacent road sections are also sufficiently long to completely contain a label. Here two road sections are called \emph{adjacent}, if they are connected by a path containing only junction edges.  
  \item Cut any road section into two halves that is sufficiently long to contain a label twice in a row. 
\end{compactenum}

\newcommand*\pct{\scalebox{.9}{\%}}

\begin{table}[htb]
\caption{Number of connected subgraphs and road sections for road networks of five cities. 
The column \emph{subgraphs} contains the number of connected subgraphs into which the graph is decomposed after preprocessing: 
1.\ the total number of subgraphs,
2.\ the number of trees,
3.\ the number of subgraphs with one cycle, and
4.\ the number of subgraphs with more than one cycle.
The
column \emph{road sections} contains the number of road sections
1.\ in total, 2.\ matched by the preprocessing strategy, 3.\ contained in trees,
4.\ contained in subgraphs with one cycle and 
5.\ contained in subgraphs with more than one cycle. 
   }
\label{table:properties}
\centering
\begin{tabular}{ccccccccccc}
     \toprule
\multirow{2}{*}{\textbf{Number of}}     && \multicolumn{4}{c}{subgraphs (after preprocessing)} & \multicolumn{5}{c}{road sections}\\
\cmidrule(lr){2-6}
\cmidrule(lr){7-11}
 &&  total & trees & 1 cycle &  $\geq 2$ cycles & total& matched & trees &1 cycle & $\geq 2$ cycles\\
\midrule
\multirow{2}{*}{Berlin}
 && $5702$	&$4853$	&$549$	&$300$	&$49773$	&$36021$	&$8220$	&$2170$	&$3362$\\
 && $100$\pct	&$85.1$\pct	&$9.6$\pct	&$5.3$\pct	&$100$\pct	&$72.4$\pct	&$16.5$\pct	&$4.4$\pct	&$6.8$\pct\\
\midrule
\multirow{2}{*}{Paris}
 && $22929$	&$20604$	&$1742$	&$583$	&$145971$	&$81305$	&$48009$	&$8329$	&$8328$\\
 && $100$\pct	&$89.9$\pct	&$7.6$\pct	&$2.5$\pct	&$100$\pct	&$55.7$\pct	&$32.9$\pct	&$5.7$\pct	&$5.7$\pct\\
\midrule
\multirow{2}{*}{London}
 && $21825$	&$20538$	&$1012$	&$275$	&$143856$	&$91405$	&$44845$	&$4485$	&$3121$\\
 && $100$\pct	&$94.1$\pct	&$4.6$\pct	&$1.3$\pct	&$100$\pct	&$63.5$\pct	&$31.2$\pct	&$3.1$\pct	&$2.2$\pct\\
\midrule
\multirow{2}{*}{Los Angeles}
 && $48248$	&$47131$	&$767$	&$350$	&$397505$	&$268334$	&$113842$	&$5149$	&$10180$\\
 && $100$\pct	&$97.7$\pct	&$1.6$\pct	&$0.7$\pct	&$100$\pct	&$67.5$\pct	&$28.6$\pct	&$1.3$\pct	&$2.6$\pct\\
\midrule
\multirow{2}{*}{New York City}
 && $10318$	&$9817$	&$306$	&$195$	&$108417$	&$72057$	&$25549$	&$3011$	&$7800$\\
 && $100$\pct	&$95.1$\pct	&$3$\pct	&$1.9$\pct	&$100$\pct	&$66.5$\pct	&$23.6$\pct	&$2.8$\pct	&$7.2$\pct\\
\bottomrule
\end{tabular}
\end{table}

That preprocessing strategy decomposed the road graphs into a large
number of subgraphs; see Table~\ref{table:properties}. For example,
for the road network of London, which consists of $143856$ road
sections, the rules of the preprocessing strategy matched $91405$ road
sections, so that the road graph decomposed into $21825$
subgraphs. Note that if we are able to label those subgraphs
optimally, we obtain an optimal labeling for the whole road network by
the choice of the preprocessing rules. Table~\ref{table:properties}
further shows that most of those subgraphs are trees ($85.1\pct$ for
Berlin as a minimum and $97.7\pct$ for Los Angeles as a
maximum). Hence, using our tree labeling algorithm we can label a
large number of the remaining subgraphs optimally. We conjecture that
using the preprocessing strategy in combination with the tree labeling
algorithm and some heuristics or exact methods for the non-tree
subgraphs we can label real-world instances near-optimally. This
hypothesis is also supported by the observation that most of the road
sections are either matched by the preprocessing strategy or are
contained in trees ($55.7\pct+32.9\pct=88.6\pct$ for Paris as a
minimum and $67.5\pct+28.6\pct=96.1\pct$ for Los Angeles as a
maximum).  For our planned companion paper we are currently working on
corresponding experiments investigating that conjecture. Further, we
are developing heuristics and exact algorithms for labeling the
remaining non-tree subgraphs.

For example we can improve our results by adapting our tree
labeling algorithm to subgraphs containing exactly one cycle $C$. We
observe that there are three cases for such a subgraph:
\begin{inparaenum}[(1)]
 \item no label identifies any road section of $C$,
 \item there is a label~$\ell$ that identifies only road sections of
   $C$, or
 \item there is a label~$\ell$ that identifies both road sections of
   $C$ and road sections of the remaining component.
\end{inparaenum}
In the first case we can remove $C$ completely from the
subgraph, such that it decomposes into a set of trees. In the second and third
case the label~$\ell$ splits the cycle $C$ so that the remaining road
sections form trees. We explore all choices of $\ell$ taking the best
choice. Hence, we can label subgraphs containing exactly one cycle optimally,
which further increases the number of optimally labeled subgraphs
($92.8\pct$ for New York City as a minimum and $97.8\pct$ for London as a
maximum).

\section{Conclusions and Outlook}\label{sec:conclusion}
In this paper we investigated the problem of maximizing the number of
identified road sections in a labeling of a road map; we showed that
it is \NP-hard in general, but can be solved in $O(n^3)$ time and
linear space for the special case of trees.

The underlying road graphs of real-world road maps are rarely
trees. Initial experimental evidence indicates, however, that road maps can be
decomposed into a large number of subgraphs by placing trivially
optimal road labels and removing the corresponding edges from the
graph. It turns out that between 85.1\% and 97.7\% of the resulting
subgraphs are actually trees, which we can label optimally by our
proposed algorithm. As a consequence, this means that a large fraction
(between 88.6\% and 96.1\%) of all road sections in our real-world road
graphs can be labeled optimally by combining this simple preprocessing
strategy with the tree labeling algorithm.  We are investigating
further heuristic and exact approaches for labeling the remaining
non-tree subgraphs (e.g., by finding suitable spanning trees and
forests) for a separate companion paper.



\newpage

\newpage
\appendix

\section{Computational Complexity}

\subsection{Description of an Alternative Clause Gadget}
\label{apx:alt-clause}

In this section we describe a clause gadget that can be used as an
alternative to the one presented in Section~\ref{sec:np-hardness}. Since it
consists only of roads that are paths, this gadget strengthens Theorem~\ref{thm:npc}.

\begin{theorem}\label{thm:npc_stronger}
  For a given road map $\mathcal M$ and an integer $K$ it is \NP-hard to decide if in total at least $K$ road sections can be identified, even if all roads are paths.
\end{theorem}

The clause gadget consists of ten roads, $r$, $g_a$, $g_b$, $g_c$
,$a_i$, $b_i$ and $c_i$ with~$i\in\{1,2\}$ that all are paths; see
Fig.~\ref{fig:np:clause2}.  Going along~$r$ from one end to the other,
the junctions with the roads $a_i$, $b_i$ and $c_i$ ($1\leq i \leq 2$)
occur in three densely packed blocks.  The blocks are described by the
sequence of roads intersecting~$r$. The first block
is~$B_a=(a_1,c_2,b_1,a_2)$, the second block is
$B_b=(a_2,b_1,c_1,b_2)$ and the third block is
$B_c=(b_2,c_1,c_2,a_1)$. The label length of~$r$ is chosen so that at
most three labels can be placed on $r$, but each road section is
shorter than a label of~$r$. Choosing the length of the road sections
appropriately, we further ensure that we can place a label that
crosses all junctions of one of the blocks without crossing the
junctions of another block.

We now describe junctions of the roads~$g_a$, $g_b$, $g_c$ ,$a_i$,
$b_i$ and $c_i$ with~$i\in\{1,2\}$.  The road~$a_1$ first intersects
$g_a$ and then $r$ twice. Let~$s^1_{a_1}$, $s^2_{a_1}$, $s^3_{a_1}$
and~$s^4_{a_1}$ denote these road sections in that particular
order. The length of~$s^1_{a_1}$ is chosen so that a single label can
be placed on~$s^1_{a_1}$, while the others are shorter than the
label length of~$a_1$. More specifically, we define $a_1$'s label
length such that a label identifies the sections in either
$\{s_{a_1}^1\}$, $\{s_{a_1}^1,s_{a_1}^2\}$,
$\{s_{a_1}^1,s_{a_1}^2,s_{a_1}^3\}$,
$\{s_{a_1}^2,s_{a_1}^3,s_{a_1}^4\}$ or $\{s_{a_1}^3,s_{a_1}^4\}$. We
define the intersections and the label length for~$a_2$,
analogously. Further,~$g_a$ intersects $a_1$ and $a_2$ in one
junction, i.e., the edge of~$g_a$ connecting both junction vertices is
a junction edge. The label length of~$g_a$ is chosen so that a label
can cross~$g_a$'s only junction. The length of~$g_a$'s road sections
is at least as long as $g_a$'s label length.  We call $g_a$ a
\emph{gate}, because later this road will be connected to the end road
of a chain by a junction; see violet square in
Fig.~\ref{fig:np:clause2:structure}. For $b_1$, $b_2$, $c_1$, $c_2$ we
introduce analogous junctions and road sections, however, $b_1$ and $b_2$
intersect~$g_b$ instead of~$g_a$, and $c_1$ and $c_2$ intersect~$g_c$
instead of~$g_a$.

\newcommand{\scaleClause}{0.9}

\begin{figure}[tb]
  \centering \subfigure[]{
    \includegraphics[page=8,scale=\scaleClause]{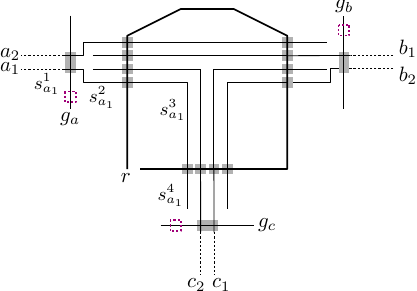}
    \label{fig:np:clause2:structure}
  } \subfigure[]{
    \includegraphics[page=9,scale=\scaleClause]{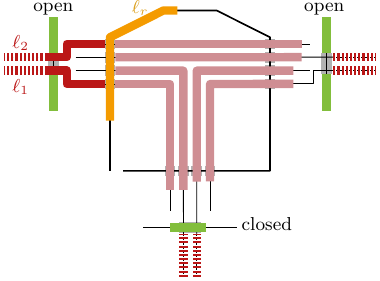}
    \label{fig:np:clause2:true}
  } \subfigure[]{
    \includegraphics[page=10,scale=\scaleClause]{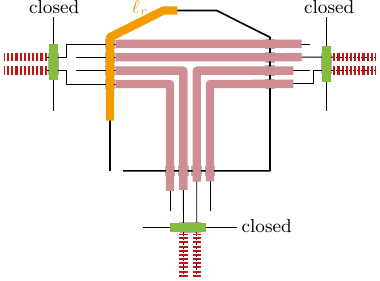}
    \label{fig:np:clause2:false}
  }
  \caption{Illustration of alternative clause gadget, which only uses
    paths as roads. \protect\subref{fig:np:clause2:structure}
    Structure of the clause
    gadget.\protect\subref{fig:np:clause2:true} Optimal labeleling for
    the case that at least one literal is
    \emph{true}.\protect\subref{fig:np:clause2:false} Optimal labeling
    for the case that all literals are \emph{false}. }
  \label{fig:np:clause2}
\end{figure}

In order to identify both road sections of a gate, either two labels
can be placed on the road sections separately, or one label that goes
through the junction.  In the former case the gate is \emph{open} and
in latter case it is \emph{closed}; see
Fig.~\ref{fig:np:clause2:true}.  We observe that it only makes sense
to close a gate, if at least one road section of the gate does not
allow to place a label that is only contained in that road section.
This case will occur if and only if the connected chain
transmits the value \emph{false} to the clause.

Assume that at least one gate is open, i.e., one literal of the clause
is true; see Fig.~\ref{fig:np:clause2:true}. Without loss of
generality let~$g_a$ be open. We place a label~$\ell_r$ on $r$ such
that it crosses the junctions of block $B_a$ and identifies 5
sections. Since~$g_a$ is open, we can place a label~$\ell_1$ that
identifies~$s^1_{a_1}$ and $s^2_{a_1}$. Analogously, we can place a
label~$\ell_2$ identifying $s^1_{a_2}$ and $s^2_{a_2}$. Placing further
labels as indicated in Fig.~\ref{fig:np:clause2:true}, we identify five road
sections of~$r$ and all road sections of any other road except for~$s^4_{c_2}$, $s^4_{b_1}$.
Hence, 33 road sections are identified.

We observe that we can place the labels of $b_1$, $b_2$, $c_1$, $c_2$
such that they do not cross the junctions of~$g_b$ and $g_c$,
respectively. Hence, it does not matter whether $g_b$ and $g_c$ are
closed or open, i.e., it does not matter whether the corresponding 
literals are \emph{true} or \emph{false}.

We now argue that this is an optimal labeling. If~$s^4_{c_2}$ or
$s^4_{b_1}$ were labeled, the label $\ell_r$ must be placed such that
the junctions of $r$ with~$c_2$ and $b_1$ are not crossed,
respectively. This decreases the number of identified road sections as
least as much identifying $s^4_{c_2}$ and $s^4_{b_1}$ increases the
number of identified road sections. In order to identify at least one
of the unidentified road sections of~$r$, we need to place a label
that crosses~$B_b$ or $B_c$. Obviously, this yields a smaller number
of identified road sections than 31.

Finally, assume that all gates are closed; see
Fig.~\ref{fig:np:clause2:false}. Consider, the same labeling as
before. However, this time we cannot label~$s^2_{a_1}$ and $s^2_{a_2}$
anymore. Hence, this labeling has only 29 identified road sections. Obviously,
it cannot be improved by changing the placement of the remaining labels or adding 
labels.


\begin{thebibliography}{10}
\small
\bibitem{street-name-placement}
F.~Chiri\'{e}.
\newblock Automated name placement with high cartographic quality: City street
  maps.
\newblock {\em Cartography and Geo. Inf. Science}, 27(2):101--110, 2000.

\bibitem{Edmondson96}
S.~Edmondson, J.~Christensen, J.~Marks, and S.~M. Shieber.
\newblock A general cartographic labelling algorithm.
\newblock {\em Cartographica}, 33(4):13--24, 1996.

\bibitem{imhof}
E.~Imhof.
\newblock Positioning names on maps.
\newblock {\em Amer. Cartogr.}, pages 128--144, 1975.

\bibitem{l-pftu-82}
D.~Lichtenstein.
\newblock Planar formulae and their uses.
\newblock {\em SIAM J. Comput.}, 11(2):329--343, 1982.

\bibitem{Maass07}
S.~Maass and J.~D\"{o}llner.
\newblock Embedded labels for line features in interactive 3d virtual
  environments.
\newblock In {\em Proc. 5th Int. Conf. Computer Graphics, Virtual Reality,
  Visualisation and Interaction in Africa}, AFRIGRAPH '07, pages 53--59. ACM,
  2007.

\bibitem{downtownLabeling}
G.~Neyer and F.~Wagner.
\newblock Labeling downtown.
\newblock In {\em Algorithms and Complexity (CIAC'00)}, volume 1767 of {\em
  LNCS}, pages 113--124. Springer, 2000.

\bibitem{criteria}
A.~Reimer and M.~Rylov.
\newblock Point-feature lettering of high cartographic quality: A
  multi-criteria model with practical implementation.
\newblock In {\em EuroCG'14}, Ein-Gedi, Israel, 2014.

\bibitem{labelingManhattan}
S.~Seibert and W.~Unger.
\newblock The hardness of placing street names in a {Manhattan} type map.
\newblock {\em Theor. Comp. Sci.}, 285:89--99, 2002.

\bibitem{strijk2001}
T.~Strijk.
\newblock {\em Geometric Algorithms for Cartographic Label Placement}.
\newblock Dissertation, Utrecht University, 2001.

\bibitem{Vaaraniemi12}
M.~Vaaraniemi, M.~Treib, and R.~Westermann.
\newblock Temporally coherent real-time labeling of dynamic scenes.
\newblock In {\em Proc. 3rd Int. Conf. Comput. Geospatial Research Appl.},
  COM.Geo '12, pages 17:1--17:10. ACM, 2012.

\bibitem{overview}
M.~van Kreveld.
\newblock Geographic information systems.
\newblock In {\em Handbook of Discrete and Computational Geometry, Second
  Edition}, chapter~58, pages 1293--1314. CRC Press, 2010.

\bibitem{wkksa-seahq-00}
A.~Wolff, L.~Knipping, M.~van Kreveld, T.~Strijk, and P.~K. Agarwal.
\newblock A simple and efficient algorithm for high-quality line labeling.
\newblock In {\em Innovations in GIS VII: GeoComputation}, chapter~11, pages
  147--159. Taylor \& Francis, 2000.

\bibitem{bibliography}
A.~Wolff and T.~Strijk.
\newblock The map labeling bibliography.
\newblock \url{
  http://liinwww.ira.uka.de/bibliography/Theory/map.labeling.html}, 2009.

\end{thebibliography}
\end{document}